\documentclass[12pt,a4paper]{amsart}
\usepackage{amssymb}
\usepackage{multirow}
\usepackage{graphicx}
\usepackage{amsmath,amsbsy,amsfonts,amsthm,latexsym,amsopn,amstext,amsxtra,euscript,amscd}

\theoremstyle{plain}
\newtheorem{theorem}{Theorem}[section]

\theoremstyle{definition}

\theoremstyle{remark}

\numberwithin{equation}{section}

\numberwithin{table}{section}

\numberwithin{figure}{section}

\begin{document}

\title[Eigenpairs of a Sylvester-Kac type matrix]{The eigenpairs of a Sylvester-Kac type matrix associated with a simple model for one-dimensional deposition and evaporation}


\author{C.M. da Fonseca}
\address{Departamento de Matem\'atica, Universidade de Coimbra, 3001-501 Coimbra, Portugal}
\email{cmf@mat.uc.pt}

\author{Dan A. Mazilu}
\address{Department of Physics and Engineering, Washington and Lee University, Lexington, VA 24450, USA}
\email{mazilud@wlu.edu}

\author{Irina Mazilu}
\address{Department of Physics and Engineering, Washington and Lee University, Lexington, VA 24450, USA}
\email{mazilui@wlu.edu}

\author{H. Thomas Williams}
\address{Department of Physics and Engineering, Washington and Lee University, Lexington, VA 24450, USA}
\email{williamsh@wlu.edu}


\date{\today}


\begin{abstract}
A straightforward model for deposition and evaporation on discrete cells of a finite array of any dimension leads to a matrix equation involving a Sylvester-Kac type matrix.
The eigenvalues and eigenvectors of the general matrix are determined for an arbitrary number of cells.  A variety of models to which this solution may be applied are discussed.
\end{abstract}

\maketitle

\section{Introduction:  the physical model}

 Consider a set of $n$ cells, arranged on a $D$-dimensional lattice. Each cell of the lattice has two states, empty or filled:  empty cells are filled at a rate $\alpha$; filled cells are emptied at a rate $\beta$.
Let $Q_k$ represent the time-dependent ensemble average probability that $k$ cells of the lattice are filled. This satisfies a rate equation:
\begin{equation}
\frac{d}{dt}\,  Q_k = -((n-k) \alpha + k \beta) Q_k + (n-k+1) \alpha Q_{k-1} + (k+1) \beta Q_{k+1} . \label{rateeq}
\end{equation}

This random sequential model is quite general and versatile, and can be customized to describe a variety of
two-state physical systems that exhibit adsorption and evaporation processes.  One-dimensional sequential adsorption models have been studied thoroughly in different physics contexts \cite{E1993,P1997}. Adsorption in two dimensions is not as well understood, however.
There are quite a few computational papers \cite{E1993} on the topic, but few analytical solutions exist for the general two-dimensional case. The adsorption of particles is exactly solvable in higher dimensions
only for  tree-like lattices.  Recently, analytical results have been reported for the random sequential process \cite{CAP2007} and reaction-diffusion processes on Cayley trees and Bethe lattices \cite{AO2010,GS2008,MAK2007}.
The standard method used to study these systems is the \emph{empty-interval method} \cite{KRB2010}. This mathematical method fails when evaporation is considered. We here demonstrate that a matrix theory approach can lead to exact results for a variety of physical systems.

Two specific experimental topics motivate our paper. One is the self-assembly  mechanism of charged nanoparticles on a glass substrate \cite{I1966}. Known in literature as ionic self-assembled monolayers (ISAM), this technique has been used successfully in making antireflective coatings \cite{dVEH2004,YZHR2006}. Physical properties of these coatings depend upon the surface coverage of the substrate.
The deposition process is stochastic,  with particles attaching to and detaching from the substrate, so a random sequential adsorption model with evaporation is appropriate.

The second motivating experimental setting involves properties of  synthetic polymers called \emph{dendrimers}, which have potential use as a drug delivery mechanism via drug encapsulation \cite{KMKKL2003}.
Dendrimers are physical analogs of Cayley tree structures. They are highly branched, spherical polymers that consist of hydrocarbon chains with various functional groups attached to a central core molecule.
The precise control that can be exerted over their size, molecular architecture, and chemical properties give dendrimers great potential in the pharmaceutical industry as effective carriers for drug molecules.
The attachment and release of the drug molecules is a stochastic process,  and so also can be modeled by a general random sequential  model as discussed herein.

The random sequential model can also address other kinds of problems as diverse as voting behavior and the spread of epidemics \cite{CFL2009}. Epidemic-type models abound in the literature  \cite{DH2000,M2002}, from simple ones that capture only the basic rules of the infection mechanism, to complex models that account for spatial spread, age structure and the possibility of immunization. Epidemic models have been applied successfully in other fields including the social sciences (voter models, rumor spreading models)  and computer science (computer virus propagation in a network) \cite{L1999,PA2009}.

\section{Mathematical model}

The rate equation \eqref{rateeq} easily transforms into an equation for the $n+1$ dimensional vector $Q$ with components $Q_0, Q_1, \ldots , Q_n$ as
$$ \frac{dQ}{dt} = M Q ,$$
in which $M$ is the tridiagonal square matrix
\begin{equation}
   M = \left( \begin{array}{ccccccc}
     -n\alpha & \beta &&&&&\\
     n\alpha & -(n-1)\alpha - \beta & 2\beta &&&& \\
     & (n-1)\alpha & -(n-2)\alpha-2\beta &  3\beta &&& \\
     && (n-2)\alpha & \ddots & \ddots  && \\
     &&& \ddots & \ddots & n\beta &\\
     &&&& \alpha & -n\beta & \\
\end{array} \right) \label{M} .
\end{equation}
The general time-dependent solution for $Q$ is a linear combination of terms of the form $u_k \exp(\lambda_k t)$ where $u_k$ is the $k$-th eigenvector of $M$ and $\lambda_k$ is the corresponding eigenvalue.

There is a long history of interest in and work on the eigenvalue problem for matrices similar to $M$.  The $(n+1)\times (n+1)$ Sylvester-Kac matrix is a tridiagonal matrix with zero main diagonal, superdiagonal $(1,2,\ldots,n)$, and subdiagonal $(n,\ldots,2,1)$, i.e.,
$$
\left(
  \begin{array}{cccccc}
    0 & 1 &&&& \\
    n & 0 & 2 &&& \\
    & n-1 & \ddots & \ddots && \\
    && \ddots & \ddots & n-1 & \\
    &&& 2 & 0 & n \\
    &&&& 1 & 0 \\
  \end{array}
\right) \, .
$$
This matrix and its characteristic polynomial were first presented by Sylvester in a short communication in $1854$ \cite{S1854}. It was conjectured that its spectrum is $2k-n$, for $k=0,1,\ldots,n$.  The first proof of Sylvester's determinant formula is attributed by Muir to Francesco Mazza in $1866$ \cite[p.442]{M1923} 
(for another interesting similar matrix, the reader is referred to \cite[pp. 432-434]{M1923}).  However, we point out that there is a typo in the book of Muir. In fact, the correct formula should be read as
$$
\left|
  \begin{array}{ccccc}
    \lambda & a_1 &&& \\
    a_n & \lambda & a_2 && \\
    & a_{n-1} & \ddots & \ddots & \\
    && \ddots & \ddots & a_n \\
    &&& a_1 & \lambda \\
  \end{array}
\right|=
\left|
  \begin{array}{cc}
    \lambda & a_n \\
    a_n & \lambda \\
  \end{array}
\right|
\left|
  \begin{array}{cccc}
   \lambda & a_1 &&\\
   a_{n-2} & \ddots & \ddots & \\
   & \ddots & \ddots & a_{n-2} \\
   && a_1 & \lambda \\
  \end{array}
\right|\, ,
$$
where ``$a_1,a_2,\ldots,a_n$ are elements increasing by the common difference $a_1$".

Only many decades
after, these matrices and their eigenvalues were subject of study in a foundational paper of Schr\"odinger \cite{S1926}, but without a proof (see \cite{BR2007,TT1991}). It appears that Mark Kac \cite{K1947} in
$1947$ was in fact the first to prove the formula, using the method of generating functions, and to provide a polynomial characterization of the eigenvectors. 
Results on the spectrum of this matrix were independently rediscovered and extended by R\'ozsa
\cite{R1957}, Clement \cite{C1959}, Vincze \cite{V1964}, Taussky and Todd \cite{TT1991}, and
Edelman and Kostlan \cite{EK1994} based on different approaches.

Applications of this simply-structured tridiagonal matrix range from
linear algebra, orthogonal polynomials, numerical analysis, graph theory, statistics,
to physical models and biogeography
\cite{A2009,A2005,BCN1989,BCD2007,CE2012,EK1995,EK1994,F2007,H1998,IS2011,I2002,MT2005,R1957}.

One of the most important non-trivial and interesting extensions of the Sylvester-Kac matrices is the tridiagonal matrix for Krawtchouk polynomials, defined by
\begin{small}
\begin{equation} K(p,n) =
   \left( \begin{array}{ccccccc}
     -pn &  pn &&&&&\\
     1-p & -pn-(1-2p) & p(n-1) &&&& \\
     & 2(1-p) & -pn-2(1-2p) &  p(n-2) &&& \\
     && 3(1-p) & \ddots & \ddots  && \\
     &&& \ddots & \ddots & p &\\
     &&&& n(1-p) & -(1-p)n & \\
\end{array} \right) .
\label{K}
\end{equation}
\end{small}

In 2005, using distinct techniques, Richard Askey \cite{A2005} and Olga Holtz
\cite{H2005} obtained simultaneously the eigenvalues for $K(p,n)$ and, later, Chu and Wang \cite{CW2008}
described the associated eigenvectors. 

\begin{theorem}[\cite{A2005,CW2008,H2005}]
The characteristic polynomial of $K(p,n)$ defined in \eqref{K} is
$$
\det (K(p,n)-xI_{n+1})=\prod_{k=0}^{n}(x+k)
$$
and, for $k=0,1,\ldots,n$ the eigenvector $u_k=(u_{k,0},u_{k,1},\ldots,u_{k,n})$ associated with eigenvalue $-k$,
has elements
$$
u_{k,\ell}=\sum_{j=0}^{\min\{\ell,k\}}(-1)^{\ell-j}\binom{\ell}{j}\binom{n-j}{k-j}p^{-j}\, ,
$$
for $\ell=0,1,\ldots,n$.
\end{theorem}

The matrix $M$ that we study in this note contains, as a particular case, the matrix $K(p,n)$ by transposition.

\section{Eigenpairs for $M$}

The equilibrium (long-time) behavior of the probabilities $Q_k$ are components of the eigenvector of $M$ corresponding to eigenvalue zero.  These static values of $Q_k$ can easily be extracted directly from the rate equation \eqref{rateeq} by setting  $dQ_k / dt = 0$.  This yields a three-term recursion relation for $Q_k$ that begins at $k=1$ as a two-term expression.  Setting $Q_0=1$ (arbitrary, and possible as long as $\alpha \neq 0$) allows us to evaluate each successive $Q_k$:
\[ Q_k = \eta^{-k} \frac{n!}{k!(n-k)!}\,  , \]
where $\eta:= \beta / \alpha$.

The relative probabilities $Q_k$, when divided by the sum
\[ \sum_{k=0}^{n} Q_k = \sum_{k=0}^{n} \eta^{-k} \frac{n!}{k!(n-k)!} = \left(\frac{1+\eta}{\eta}\right)^n ,\]
produce absolute probabilities
\[ Q'_k = \frac{\eta^{n-k}}{(1+\eta)^{n}} \frac{n!}{k!(n-k)!}\, . \]
From these we can also calculate the average coverage:
\[ \sum_{k=0}^{n} k\, Q'_k = \sum_{n=0}^n\frac{\eta^{n-k}}{(1+\eta)^n} \frac{n!}{(k-1)!(n-k)!} = n \frac{1}{1+\eta} =  n \frac{\alpha}{\alpha+\beta}\, , \]
confirming a result easily derived as well from mean-field analysis.

More generally, the full spectrum of eigenvalues and corresponding eigenvectors for the Sylvester-Kac matrix defined in \eqref{M} is as follows:

\begin{theorem} \label{main}
The eigenvalues of the matrix $M$ of order $n+1$  defined in \eqref{M} are
$$
\lambda_k = -k(\alpha + \beta)\, ,
$$
for $k=0,1,\ldots,n$; and $u_k=(u_{k,0},u_{k,1},\ldots,u_{k,n})$ is the eigenvector associated with $\lambda_k$,
with
\begin{equation} \label{ukl}
u_{k,\ell}=\sum_{j=\ell-k}^{n-k}  (-)^{k+\ell+j} \eta^{n-k-j} \binom{k}{\ell-j} \binom{n-k}{j}\, ,
\end{equation}
for $\ell=0,1,\ldots,n$, where $\eta \equiv \beta/\alpha$.
\end{theorem}

\begin{proof}
Our approach consists in checking the eigenvalue equation for every case.
For the first row of  $M$, we require
\[ -n\alpha u_{k,0} + \beta u_{k,1} = -k (\alpha + \beta) u_{k,0} \]
for all $k$. Dividing the equation by $\alpha$ and moving all terms to the left gives
\[ (-n+k+k \eta) u_{k,0} + \eta u_{k,1} = 0 . \]
Utilizing the explicit expressions for the $u_{k,\ell}$'s from \eqref{ukl} leads simply to
\[  (-n+k+k\eta)   (-)^{k} \eta^{n-k} -   (-)^{k} \eta^{n-k+1} k  + (-)^{k} \eta^{n-k} (n-k)  \]
which vanishes, confirming the eigenvalue equations per the top row of the matrix.

Similarly for the last row of $M$
\[  \alpha u_{k,n-1} - n \beta  u_{k,n} = -k (\alpha + \beta) u_{k,n} \]
for all $k$, equivalent to
\[  u_{k,n-1}+ (k +(k-n) \eta) u_{k,n} = 0\, . \]
Again utilizing explicit expressions for the $u_{k,\ell}$'s  we get only three non-vanishing terms
\[  \eta (n-k)-k+(k+(k-n) \eta) =0\, , \]
confirming the eigenvalue equations per the bottom row.

For all the other rows (e.g. row $\ell$) of $M$ the eigenvalue equations can be expressed as
\begin{equation}\label{lthrow}
(n-\ell+1) u_{k,\ell-1} +(k+\ell-n +(k-\ell)\eta) u_{k,\ell} + (\ell+1) \eta\, u_{k,\ell+1} = 0\, .
\end{equation}
For this calculation, it is useful to change the summation variable in each expression for $u_{k,\ell}$ in such a way that each term is a
sum over $q$ with the same power of $\eta$ in each term, i.e., $\eta^q$.  This transforms the first term in
\eqref{lthrow} into
\[ T_1 = -(n-\ell+1) \sum_{q=0}^{n-\ell+1} (-)^{n+\ell+q} \eta^q \binom{k}{k+\ell+q-1-n}\binom{n-k}{n-k-q}\, ; \]
the second term is best expressed as two sums (the second having a multiplicative $\eta$)
\[ T_2 = (k+\ell-n) \sum_{q=0}^{n-\ell} (-)^{n+\ell+q} \eta^q \binom{k}{k+\ell+q-n}\binom{n-k}{n-k-q} \]
and
\[ T_3 = -(k-\ell) \sum_{q=1}^{n-\ell+1} (-)^{n+\ell+q} \eta^q \binom{k}{k+\ell+q-1-n}\binom{n-k}{n+1-k-q}; \]
and the final term becomes
\[ T_4 = (\ell+1) \sum_{q=0}^{n-\ell+1} (-)^{n+\ell+q} \eta^q \binom{k}{k+\ell+q-n}\binom{n-k}{n+1-k-q}\, . \]

Because of the uneven summation limits we must look at two special cases.  The coefficient of $\eta^0$
in the sum of these four sums has contributions only from $T_1$ and $T_2$:

$$
-(n-\ell+1) (-)^{n+\ell} \binom{k}{k+\ell-1-n}\binom{n-k}{n-k}+(k+\ell-n)(-)^{n+\ell}  \binom{k}{k+\ell-n}\binom{n-k}{n-k}
$$
$$= -(n-\ell+1) (-)^{n+\ell} \frac{k!}{(k+\ell-n)!(n-\ell)!} \frac{k+\ell-n}{(n-\ell+1)} + (k+\ell-n)(-)^{n+\ell} \binom{k}{k+\ell-n} $$
easily seen to vanish, as desired.

The second special case involves the coefficient of $\eta^{n-\ell+1}$, which has contributions only from $T_1$ and $T_3$:

\noindent $\displaystyle{
-(n-\ell+1) (-)^{1} \binom{k}{k}\binom{n-k}{\ell-k-1}- (k-\ell)(-)^{1}  \binom{k}{k}\binom{n-k}{\ell-k}}$

$$= (n-\ell+1)  \binom{n-k}{\ell-k-1} + (k-\ell) \frac{(n-k)!}{(\ell-k-1)!(n-\ell+1)!} \frac{n-\ell+1}{\ell-k} $$

which also clearly vanishes.

The confirmation concludes with examination of the coefficient of $\eta^q$ ($0 < q < n-\ell+1$) in the sum of the four $T_i$ terms:
$$
  - (n-\ell+1) \binom{k}{k+\ell+q-1-n}\binom{n-k}{n-k-q}+(k+\ell-n) \binom{k}{k+\ell+q-n}\binom{n-k}{n-k-q}$$

$$- (k-\ell) \binom{k}{k+\ell+q-1-n}\binom{n-k}{n+1-k-q}+(\ell+1) \binom{k}{k+\ell+q-n}\binom{n-k}{n+1-k-q} $$

\noindent $\displaystyle{= \binom{k}{k+\ell+q-n}\binom{n-k}{n-k-q} \left(-(n-\ell+1) \, \frac{k+\ell+q-n}{n+1-\ell-q} + (k+\ell-n) \right. }$

\hfill $\displaystyle{-\left. (k-\ell)\,\frac{k+\ell+q-n}{n+1-\ell-q}\,\frac{q}{n+1-k-q} +(\ell-1)\frac{q}{n+1-k-q}  \right) .}$

\noindent Patient algebra  can confirm that the final bracketed expression vanishes, thereby completing the proof.
\end{proof}

Analytic software was used to find the eigenvalues and eigenvectors for the matrix
$M$ \eqref{M} for small values of $n$ (up to ten) and from these forms general expressions for arbitrary
$n$ were conjectured.  The method shown above was used to verify the conjectures.

We believe that Theorem \ref{main} will lead to new interesting relations to Krawtchouk, Hahn, 
and $q$-Racah polynomials as in \cite{A2005,BCD2007,CW2008,H2005}.
Some problems involving distance regular graphs can also be considered \cite{BCN1989,TT1991}. 
Moreover, the original connection with problems in statistical mechanics,
comprehensively studied by Mark Kac \cite{K1947}, can be extended.

\section{Acknowledgement}

We are grateful to the referee for the careful reading of our manuscript and the valuable comments.

\end{document}